\documentclass[11pt]{article}
\usepackage{epsfig,amsmath,amsthm,latexsym,graphicx,amsfonts,amssymb}
\newcommand{\beq}{\begin{equation}}
\newcommand{\eeq}{\end{equation}}
\newcommand{\R}{\mathbb R}

\newcommand{\E}{\mathbb E}

\newcommand{\grad}{\operatorname{\nabla}}
\newcommand{\cX}{\mathcal{X}}
\newcommand{\htheta}{\hat{\theta}}
\newcommand{\hmu}{\hat{\mu}}
\newcommand{\thetad}{\dot{\theta}}

\newcommand{\mud}{\dot{\mu}}

\newcommand{\qd}{\dot{q}}
\newcommand{\qdd}{\ddot{q}}
\newcommand{\Id}{\dot{I}_1}
\newcommand{\Idd}{\ddot{I}_1}
\newcommand{\defeq}{\stackrel{\mbox{\tiny{def}}}{=}}
\newtheorem{theorem}{Theorem}
\newtheorem{corol}[theorem]{Corollary}
\newtheorem{lemma}[theorem]{Lemma}

\newtheorem{remark}{Remark}

\begin{document}

\title{SMML estimators for exponential families \\ with continuous sufficient statistics}
\author{James G. Dowty}
\date{\today}

\maketitle

\abstract{The minimum message length principle is an information theoretic criterion that links data compression with statistical inference.  This paper studies the strict minimum message length (SMML) estimator for $d$-dimensional exponential families with continuous sufficient statistics, for all $d \ge 1$.  The partition of an SMML estimator is shown to consist of convex polytopes (i.e. convex polygons when $d=2$) which can be described explicitly in terms of the assertions and coding probabilities.  While this result is known, we give a new proof based on the calculus of variations, and this approach gives some interesting new inequalities for SMML estimators.  We also use this result to construct an SMML estimator for a $2$-dimensional normal random variable with known variance and a normal prior on its mean.  }

\section{Introduction}

The minimum message length (MML) principle \cite{wallace_boulton} is an information theoretic criterion that links data compression with statistical inference \cite{wallace}.  It has a number of useful properties and it has close connections with Kolmogorov complexity \cite{wallace_dowe}.  Using the MML principle to construct estimators is known to be NP-hard in general \cite{farr} so it is common to use approximations in practice \cite{wallace}.  The term `strict minimum message length' (SMML) is used for the exact MML criterion, to distinguish it from the various approximations.

The only known algorithm for calculating an SMML estimator is Farr's algorithm \cite{farr} which applies to data taking values in a finite set which is (in some sense) $1$-dimensional.  A method for calculating SMML estimators for
$1$-dimensional exponential families with continuous sufficient statistics was also recently given in \cite{dowty}.  However, calculating SMML estimators for higher-dimensional data is still a difficult problem.

This paper studies SMML estimators for $d$-dimensional exponential families of statistical models with continuous sufficient statistics.  Section \ref{S:review} recalls the relevant definitions and fixes our notation.
Section \ref{S:deformations} shows how the expected two-part code-length $I_1$ changes as the partition is changed by a small amount, though the proof of the main technical lemma is deferred to Appendix \ref{S:proofs}.
Section \ref{S:convex} uses this calculation to prove that the partition corresponding to an SMML estimator consists of certain convex polytopes.  Section \ref{S:ineq} then proves some interesting new inequalities which follow from the requirement that the second variation of $I_1$ should be non-negative at the SMML estimator.  Section \ref{S:construction} describes then uses a general method to calculate SMML estimators and Section \ref{S:conclusion} summarizes these results.

\section{SMML estimators for exponential families}
\label{S:review}

Partly to define our notation, this section briefly recalls the relevant facts about exponential families and their SMML estimators.

Let $\cX$ and $\Theta$ be the support and natural parameter space of the exponential family (respectively) which are both are open, connected subsets of $\R^d$ with $\Theta$ convex.  For each $\theta \in \Theta$, let $f(x | \theta)$ be the probability density function (PDF) on $\cX$ given by
\begin{equation}\label{E:exp_family}
f(x | \theta) \defeq \exp(x \cdot \theta - \psi(\theta)) h(x)
\end{equation}
for any $x \in \cX$, where the dot denotes the Euclidean inner product, $h:\cX \to \R$ is a strictly positive function and $\psi:\Theta \to \R$ is determined by the condition $1 = \int_\cX f(x | \theta) dx$ for every $\theta \in \Theta$.
Let $\mu: \Theta \to \R^d$ be the function
$$\mu(\theta) \defeq \E[X \mid \theta] = \int_\cX x f(x | \theta) dx$$
which relates the natural parametrization of the exponential family to the expectation parametrization, where $\E[X \mid \theta]$ is the expectation of any random variable with PDF (\ref{E:exp_family}).  Then by a standard result for exponential families (e.g. see Theorem 2.2.1 of \cite{kass}), $\psi$ is smooth (i.e. infinitely differentiable), $\mu$ is a diffeomorphism from $\Theta$ to its image (i.e. a smooth function with a smooth inverse) and
\begin{eqnarray}
  \mu(\theta) = \, \E[X \mid \theta] &=& \grad \psi\mid _\theta \label{E:gradpsi} \\
  \mbox{Var}(X\mid\theta) &=& \mbox{Hess}(\psi) \mid_\theta \label{E:hesspsi}
\end{eqnarray}
where $\mbox{Var}(X\mid\theta)$ is the variance-covariance matrix of any random variable with PDF (\ref{E:exp_family}), $\grad \psi$ is the gradient of $\psi$ and $\mbox{Hess}(\psi)$ is the Hessian matrix of $\psi$.

Let $\pi$ be a Bayesian prior on $\Theta$ and define the {\em marginal PDF} $r$ by
$$ r(x) \defeq \int_\Theta \pi(\theta) f(x | \theta) d\theta $$
for any $x \in \cX$.  We assume $\pi$ is chosen so that the first moment of $r$ exists.

For the case considered above, an SMML estimator with $n \ge 1$ regions is defined as follows \cite{wallace}.  Suppose we are given $\htheta_1, \ldots, \htheta_n \in \Theta$ (the assertions), $q_1, \ldots, q_n \in \R$ so that $1 = q_1 + \ldots + q_n$ and each $q_i > 0$ (the coding probabilities for the assertions) and a partition $U_1, \ldots, U_n$ of $\cX$, i.e. subsets $U_1, \ldots, U_n \subseteq \cX$ so that $\cX = \cup_{i=1}^n U_i$ and $U_i \cap U_j$ is a set of Lebesgue measure $0$ for all $i \not= j$.  We also assume that each $U_i$ has non-zero measure and we will place other restrictions on each $U_i$ in Section \ref{S:deformations}.  Let $\htheta: \cX \to \Theta$ and $q:\cX \to \R$ be the functions which take the values $\htheta_i$ and $q_i$ on $U_i$ (respectively), and note that these definitions make sense except on the set of measure $0$ where two or more $U_i$ overlap.  If we discretize the data space $\cX$ to a lattice then there is a $2$-part coding of the data which has expected length
\begin{equation}\label{E:I1_defn}
I_1 = - \E[ \log(q(X) f(X | \htheta(X) ))]
\end{equation}
plus a constant which only depends on the width of the lattice, where $X$ is a random variable with PDF $r$.  Then an SMML estimator with $n$ regions is a function $\htheta(x)$ which minimizes $I_1$ out of all estimators of this form.

The following lemma is a refinement for exponential families of some well-known facts about SMML estimators.

\begin{lemma}
\label{L:SMML_basic}
If an SMML estimator has partition $U_1, \ldots, U_n$, assertions $\htheta_1, \ldots, \htheta_n$ and coding probabilities $q_1, \ldots, q_n$ then
\begin{eqnarray}
q_i &=& \int_{U_i} r(x) dx \label{E:q} \\
\htheta_i &=& \mu^{-1}\left( \frac{1}{q_i} \int_{U_i} x r(x) dx \right)  \label{E:htheta}
\end{eqnarray}
for each $i = 1, \ldots, n$.

Also, for any partition $U_1, \ldots, U_n$, not necessarily corresponding to an SMML estimator, if $q_i$ and $\htheta_i$ are as in (\ref{E:q}) and (\ref{E:htheta}) then
\begin{equation}\label{E:I1_finitesum}
I_1 = C - \sum_{i=1}^n  q_i \left( \log q_i + \htheta_i \cdot \mu(\htheta_i) - \psi(\htheta_i) \right)
\end{equation}
where $C$ is the constant $- \int_\cX r(x) \log h(x) dx$.
\end{lemma}

\begin{proof}
From (\ref{E:I1_defn}) we have
\begin{eqnarray}
I_1 &=&  - \int_\cX r(x) \log(q(x) f(x| \htheta(x) )) dx \nonumber \\
&=& - \sum_{i=1}^n  \int_{U_i} r(x) \log(q(x) f(x| \htheta(x) )) dx  \nonumber \\
&=& - \sum_{i=1}^n  \int_{U_i} r(x) \log(q_i f(x| \htheta_i )) dx  \nonumber \\
&=& C - \sum_{i=1}^n  \int_{U_i} r(x) \left( \log q_i +  x \cdot \htheta_i - \psi(\htheta_i) \right) dx \mbox{ by (\ref{E:exp_family})} \nonumber \\
&=& C - \sum_{i=1}^n  \left[ \left(\log q_i - \psi(\htheta_i)\right)\int_{U_i} r(x) dx
+ \htheta_i \cdot \int_{U_i} x r(x) dx  \right].  \label{E:I1_partial_sum}
\end{eqnarray}

Now, assume $U_1, \ldots, U_n$, $\htheta_1, \ldots, \htheta_n$ and $q_1, \ldots, q_n$ correspond to an SMML estimator, i.e. they represent a global minimum of $I_1$.  Then it is not possible to reduce $I_1$ by changing $q_1, \ldots, q_n$ so that $1 = q_1 + \ldots + q_n$ and with $U_1, \ldots, U_n$ and $\htheta_1, \ldots, \htheta_n$ fixed.  So by the method of Lagrange multipliers, at the SMML estimator the gradient of (\ref{E:I1_partial_sum}) (with only $q_1, \ldots, q_n$ varying) should be proportional to the gradient of the constraint function $q_1 + \ldots + q_n - 1$, i.e. there is some $\lambda \in \R$ so that for all $i = 1, \ldots, n$,
$$ \lambda = \lambda \frac{\partial \,}{\partial q_i} (q_1 + \ldots + q_n - 1)
= \frac{\partial I_1}{\partial q_i} = -\frac{1}{q_i} \int_{U_i} r(x) dx,$$
where the last step is by (\ref{E:I1_partial_sum}), so this and the condition $1 = q_1 + \ldots + q_n$ imply (\ref{E:q}).  Similarly, $I_1$ cannot be reduced by changing $\htheta_i$ while keeping $U_1, \ldots, U_n$, $\htheta_1, \ldots,\htheta_{i-1},\htheta_{i+1},\ldots, \htheta_n$ and $q_1, \ldots, q_n$ fixed. So if $\htheta_i \in \Theta \subseteq \R^d$ has co-ordinates $\htheta_i = (\htheta_{i1}, \ldots, \htheta_{id})$ then by (\ref{E:q}) and (\ref{E:I1_partial_sum}), for every $j = 1, \ldots, d$,
$$ 0 = \frac{\partial I_1}{\partial \htheta_{ij}}
= q_i \frac{\partial \,}{\partial \htheta_{ij}} \psi(\htheta_i) - \int_{U_i} x_j r(x) dx $$
so (\ref{E:htheta}) follows from (\ref{E:gradpsi}).  Lastly, (\ref{E:I1_finitesum}) follows from (\ref{E:q}), (\ref{E:htheta}) and (\ref{E:I1_partial_sum}).
\end{proof}

\section{Deformations of the partition}
\label{S:deformations}

By Lemma \ref{L:SMML_basic}, finding an SMML estimator is equivalent to finding a partition of $\cX$ which minimizes (\ref{E:I1_finitesum}) when $q_i$ and $\htheta_i$ are as in (\ref{E:q}) and (\ref{E:htheta}).  In this section, we consider an estimator defined by a partition $U_1, \ldots, U_n$ and we calculate how $I_1$ varies as we change the partition by a small amount.  This is interesting because, to first order, $I_1$ should not change under any small deformation when $U_1, \ldots, U_n$ corresponds to an SMML estimator.

We now place some fairly mild restrictions on the partitions that we consider by assuming each $U_i$ is a (not necessarily connected) $d$-manifold in $\R^d$ with a piecewise smooth boundary $\partial U_i$ (see \S3.1 of \cite{agricola} for a general description of manifolds in $\R^d$).
This means that each $U_i$ is the solid region in $\R^d$ bounded by a $(d-1)$-dimensional set $\partial U_i$ which locally has the same shape as the graph of some smooth real-valued function defined on a small ball in $\R^d$, except on a $(d-2)$-dimensional set where $\partial U_i$ is allowed to have `ridges' or `corners' like those that can occur in the graph of the minimum of a finite number of smooth (and transverse) functions.  We therefore allow each $U_i$ to have a very wide range of topologies and geometries but we do not consider partitions with fractal boundaries, for instance.
Since we have already assumed that any two regions $U_i$ and $U_j$ overlap in a set of measure $0$, we require that the interiors of $U_i$ and $U_j$ are disjoint and hence that $U_i \cap U_j \subseteq \partial U_i \cap \partial U_j$.

Now, suppose that $U_1$ and $U_2$ share a `face', i.e. that $\partial U_1 \cap \partial U_2$ contains a smooth, $(d-1)$-dimensional, curvilinear disc $D$.  We will deform the partition $U_1, \ldots, U_n$ by perturbing $D$ slightly.

Let $N$ be the unit normal vector field on $D$ which points out of $U_1$ and into $U_2$, and extend $N$ in any way to a smooth vector field defined on all of $\R^d$.  Let $g : \R^d \to \R$ be any function so that $g(x)=0$ except perhaps in a closed and bounded subset $\mbox{Supp}(g)$ of $\R^d$ (the support of $g$) which is contained in $U_1 \cup U_2$ and
which only meets $\partial U_1 \cup \partial U_2$ in a subset of $D$.

For all real $t$ close to $0$, let $F_t: \R^d \to \R^d$ be the flow of the vector field $g N$, i.e. for given $x \in \R^d$, let $F_t(x)$ be the position of a particle in $\R^d$ which starts at $x$ and whose velocity at time $t$ is $g N$ evaluated at the position of the particle (see \S3.9 of \cite{agricola}).  Each $F_t$ is a diffeomorphism from $\R^d$ to itself and it is given by $F_t(x) = x + t g(x) N(x)$ to a first order in $t$, for small $t$.  If we define
$$U_i(t) \defeq F_t(U_i)$$
then $U_1(t), \ldots, U_n(t)$ is also a partition of $\cX$ for each $t$ (since $F_t$ is a diffeomorphism).  Also, $F_0$ is the identity so $U_i(0)= U_i$ for all $i=1, \ldots, n$.  Therefore we can consider $U_1(t), \ldots, U_n(t)$ to be a deformation of the partition $U_1, \ldots, U_n$.  Also, because of the restrictions on $\mbox{Supp}(g)$ above, $U_i(t) = U_i$ for all $t$ if $i \not= 1,2$ and $D$ is the only part of $\partial U_1 \cup \partial U_2$ which changes as $t$ is varied.

We now have the following key lemma.

\begin{lemma}
\label{L:cdot}
Let $i$ be $1$ or $2$ and let $c(t) \defeq \int_{U_i(t)} \rho(x) dx$ for some smooth function $\rho: \R^d \to \R$.  Then
\begin{eqnarray*}
\left. \frac{dc}{dt}\right|_{t=0} &=& \pm \int_D g \rho \, dx \\
\left. \frac{d^2c}{dt^2}\right|_{t=0} &=& \pm \int_D  g \nabla \cdot (g \rho N) dx
\end{eqnarray*}
where $\nabla \cdot (g \rho N)$ is the divergence of $g \rho N$ and both signs are positive if $i=1$ and both are negative if $i=2$.
\end{lemma}

\begin{remark}
In Lemma \ref{L:cdot} and throughout this paper, we will often denote the integral of a function $\phi(x)$ over a subset $\Omega$ of $\R^d$ by $\int_\Omega \phi \, dx$ rather than $\int_\Omega \phi(x) \, dx$.
\end{remark}

\begin{proof}[Proof of Lemma \ref{L:cdot}]
See Sections 4.1, 6.2 and 6.3 of \cite{delfour}, especially equations 4.7, 6.9, 6.14 and 6.15.  The cases $i=1$ and $i=2$ have different signs because $N$ is the outward-pointing unit normal for $U_1$ but the inward-pointing unit normal for $U_2$. See the Appendix for an alternative proof of this key lemma.
\end{proof}

The following theorem gives the first and second variations of $I_1$ corresponding to the above deformation of the partition.

\begin{theorem} \label{T:Id_Idd}
For each $t \in \R$ in a neighbourhood of $0$, let $U_1(t), \ldots, U_n(t)$ be the partition given above. Let $q_i$ and $\htheta_i$ be the functions of $t$ given by expressions (\ref{E:q}) and (\ref{E:htheta}) with $U_i(t)$ replacing $U_i$, and let $I_1$ be the function of $t$ obtained by substituting these functions into (\ref{E:I1_finitesum}).  If $\Id$ and $\Idd$ are (respectively) the first and second derivatives of $I_1$ with respect to $t$ then
\begin{equation}\label{E:Id}
- \left. \Id \right|_{t=0} = \int_D r g (\lambda_1 - \lambda_2) \, dx
\end{equation}
and
\begin{eqnarray}
-\left. \Idd \right|_{t=0}
&=& \int_D (\lambda_1 - \lambda_2) g \nabla \cdot \left( g r N \right) \, dx
   + \left( \frac{1}{q_1} + \frac{1}{q_2} \right) \left(\int_D g r \, dx\right)^2 \nonumber \\
&& + \frac{1}{q_1} \delta_1 \cdot Q_1^{-1} \delta_1 + \frac{1}{q_2} \delta_2 \cdot Q_2^{-1} \delta_2
   + \int_D  g^2 r N \cdot (\htheta_1 - \htheta_2) \, dx.  \label{E:Idd}
\end{eqnarray}
where $\lambda_i(x) = \log q_i + x \cdot \htheta_i - \psi(\htheta_i)$, $\delta_i = \int_D (x-\mu(\htheta_i)) g r \, dx$ and $Q_i$ is the Hessian of $\psi$ evaluated at $\htheta_i$ (so $Q_i$ is symmetric and positive-definite by (\ref{E:hesspsi})).
\end{theorem}

\begin{proof}
See Appendix \ref{S:proofs}.
\end{proof}

\begin{remark} \label{R:Qi_is_FIM}
Note that $Q_i^{-1}$ in Theorem \ref{T:Id_Idd} is the Fisher information matrix for the expectation parameterisation evaluated at the expectation parameter $\mu(\htheta_i)$ \cite[Theorem 2.2.5]{kass}.
\end{remark}

\section{The partition of an SMML estimator}
\label{S:convex}

We can now prove our main theorem.  While this is essentially Wallace's R1 condition \cite[p. 156]{wallace} specialized to exponential families, we give a new proof based on the calculus of variations, and this approach implies some interesting new inequalities for SMML estimators (see Section \ref{S:ineq}).

\begin{theorem}
\label{T:Ui}
If an SMML estimator has partition $U_1, \ldots, U_n$, assertions $\htheta_1, \ldots, \htheta_n$ and coding probabilities $q_1, \ldots, q_n$ then
\begin{equation}\label{E:Ui}
U_i = \{ x \in \cX \mid \lambda_i(x) \ge \lambda_j(x) \mbox{ for all $j=1,\ldots, n$} \}
\end{equation}
where $\lambda_i$ is the linear function of $x$ given by $\lambda_i(x) = \log q_i + x \cdot \htheta_i - \psi(\htheta_i)$.  In particular, each $U_i$ is a convex polytope determined by $\htheta_1, \ldots, \htheta_n$ and $q_1, \ldots, q_n$.
\end{theorem}

\begin{proof}  This proof can be skipped without compromising the reader's understanding of the rest of the paper.

As in the statement, let an SMML estimator have partition $U_1, \ldots, U_n$, assertions $\htheta_1, \ldots, \htheta_n$ and coding probabilities $q_1, \ldots, q_n$.  If we define
$$ V_i \defeq \{ x \in \cX \mid \lambda_i(x) \ge \lambda_j(x) \mbox{ for all $j=1,\ldots, n$} \}$$
then our goal is to prove that $U_i = V_i$ for all $i=1,\ldots, n$.

For $i \not= j$, let $P_{ij} \defeq \{ x \in \R^d \mid \lambda_i(x) = \lambda_j(x) \}$ and $H_{ij} \defeq \{ x \in \R^d \mid \lambda_i(x) \ge \lambda_j(x) \}$. For any $i=1,\ldots, n$, let $C_1^{(i)}, \ldots, C_{m_i}^{(i)}$ be the closures of the connected components of $\R^d \setminus \left(\cup_{j: j \not= i} P_{ij}\right)$, i.e. each $C_k^{(i)}$ is a $d$-dimensional convex polytope with boundary lying in $\cup_{j: j \not= i} P_{ij}$ but whose interior is disjoint from all of these hyperplanes.

{\em Claim 1: $U_i$ is the union of one or more of $C_1^{(i)}, \ldots, C_{m_i}^{(i)}$.}  Assume without loss of generality that $i=1$ and that $U_1$ meets $U_2$ in a $(d-1)$-dimensional face.  As in Section \ref{S:deformations}, let $U_1(t), \ldots, U_n(t)$ be a deformation of the partition $U_1, \ldots, U_n$ corresponding to some $g$, $N$ and $D$ and let $q_i$, $\htheta_i$ and $I_1$ be functions of $t$ as in Theorem \ref{T:Id_Idd}.  An SMML estimator is a global minimum of $I_1$ so $\Id = 0$ and $\Idd \ge 0$ at $t=0$ for all deformations, so in particular these relations hold for any $g$, $N$ and $D$.  But by (\ref{E:Id}), $\Id = \int_D r g (\lambda_1 - \lambda_2) \, dx$, and this integral can vanish for all $g$ only if the integrand vanishes on $D$.  So since $r>0$, we must have $\lambda_1(x) = \lambda_2(x)$ for all $x \in D$, i.e. $D$ is contained in the hyperplane $P_{12}$.

Since $D$ is an arbitrary smooth $(d-1)$-dimensional disc contained in $U_1 \cap U_2$, this shows that all of $U_1 \cap U_2$ is contained in $P_{12}$ except perhaps a set of dimension $d-2$ where $\partial U_1$ or $\partial U_2$ is not smooth or where $\partial U_1 \cap \partial U_2$ has dimension $d-2$ or less.  Therefore a dense subset of $U_1 \cap U_2$ is contained in $P_{12}$.  But $P_{12}$ is closed so this implies that $U_1 \cap U_2 \subseteq P_{12}$.  Similar comments hold with $U_2$ replaced by any $U_j$ which shares a $(d-1)$-dimensional face with $U_1$.  Therefore $\partial U_1 \subseteq \cup_{j: j \not= 1} P_{1j}$ and the claim follows.

{\em Claim 2: If $C_k^{(i)} \subseteq U_i$, $C_k^{(i)}$ meets $P_{ij}$ in a $(d-1)$-dimensional set and $C_k^{(i)} \cap P_{ij} \subseteq \partial U_i$ then $C_k^{(i)} \subseteq H_{ij}$.}
Assume, without loss of generality, that $i=1$ and $j=2$.  Let $g$, $N$ and $D$ be as above and let $C_k^{(1)}$ be as in the statement of this claim.  Since $D$ is arbitrary, we can assume additionally that $D \subseteq C_k^{(1)} \cap P_{12}$.  Since $D \subseteq P_{12}$, the vector field $N$ is constant on $D$ and $\lambda_1 = \lambda_2$ there, so by (\ref{E:Idd}), $\Idd \ge 0$ is equivalent to
\begin{eqnarray*}
0 &\ge& -\Idd \\
&=& \left( \frac{1}{q_1} + \frac{1}{q_2} \right) \left(\int_D g r \, dx\right)^2
+ \frac{1}{q_1} \delta_1 \cdot Q_1^{-1} \delta_1 + \frac{1}{q_2} \delta_2 \cdot Q_2^{-1} \delta_2
   + N \cdot (\htheta_1 - \htheta_2) \int_D  g^2 r  \, dx
\end{eqnarray*}
and hence
\begin{equation}\label{E:Ntheta}
N \cdot (\htheta_2 - \htheta_1)  \ge \left[ \int_D  g^2 r  \, dx \right]^{-1} \left(
\left[ \frac{1}{q_1} + \frac{1}{q_2} \right] \left[\int_D g r \, dx\right]^2
+ \frac{1}{q_1} \delta_1 \cdot Q_1^{-1} \delta_1 + \frac{1}{q_2} \delta_2 \cdot Q_2^{-1} \delta_2 \right).
\end{equation}
By Theorem \ref{T:Id_Idd}, $Q_1$ and $Q_2$ are positive definite, so all terms on the right-hand side of (\ref{E:Ntheta}) are non-negative for all $g$ and strictly positive for some $g$, so
\begin{equation}\label{E:ineq}
N \cdot (\htheta_2 - \htheta_1) > 0.
\end{equation}

Now, $\grad(\lambda_1(x) - \lambda_2(x)) = \htheta_1 - \htheta_2$ so $\lambda_1 - \lambda_2$ is increasing in the direction $\htheta_1 - \htheta_2$.  Since $\lambda_1 - \lambda_2 = 0$ on $D$, $\lambda_1 > \lambda_2$ locally on the side of $D$ into which $\htheta_1 - \htheta_2$ points.  But this must be the $U_1$ side of $D$, since $N \cdot (\htheta_2 - \htheta_1) > 0$ and $N$ is, by definition, the unit normal to $D$ which points out of $U_1$ and into $U_2$.
But $C_k^{(1)}$ lies entirely on one side or the other of $P_{12}$, so the fact that part of it lies on the side where $\lambda_1 \ge \lambda_2$ implies all of it does, i.e. $C_k^{(1)} \subseteq H_{12}$.

{\em Claim 3: If $V_i$ has non-zero measure then $V_i \subseteq U_i$. }
Assume $V_i$ has non-zero measure.  Note that this implies $V_i = C_k^{(i)}$ for some $k$.  For each $C_k^{(i)}$, let $\#(C_k^{(i)})$ be the number of half-spaces $H_{ij}$ (for $i$ fixed and $j$ varying) so that $C_k^{(i)} \subseteq H_{ij}$.  Then $\#(C_k^{(i)}) \le n-1$ and $\#(C_k^{(i)}) = n-1$ if and only if $C_k^{(i)}$ lies in all $H_{ij}$ for $j\not=i$, i.e. $C_k^{(i)} = \cap_{j: j \not= i} H_{ij} = V_i$.  Let $C_k^{(i)}$ have maximal $\#(C_k^{(i)})$ out of all $C_1^{(i)}, \ldots, C_{m_i}^{(i)}$ which are contained in $U_i$.  If $C_k^{(i)} = V_i$ then $V_i \subseteq U_j$ and the claim is proved so assume, in order to derive a contradiction, that $C_k^{(i)} \not= V_i$.

If $C_k^{(i)} = \cap_{j \in J} H_{ij}$ for some $J \subseteq \{ 1, \ldots, n \} \setminus \{ i \}$ then $C_k^{(i)} = \cap_{j \in J} H_{ij} \supseteq \cup_{j: j \not= i} H_{ij} = V_i$, but then $C_k^{(i)} = V_i$ since distinct $C_1^{(i)}, \ldots, C_{m_i}^{(i)}$ only overlap in sets of measure $0$.  Taking $J$ to be the set of all $j$ so that $C_k^{(i)} \cap P_{ij}$ is a $(d-1)$-dimensional face, we therefore see that there is some $j \not= i$ so that $F \defeq C_k^{(i)} \cap P_{ij}$ is a $(d-1)$-dimensional face of $C_k^{(i)}$ but $C_k^{(i)}$ is not contained in $H_{ij}$.  By Claim 2, $F$ cannot be contained in $\partial U_i$, so $C_l^{(i)} \subseteq U_i$ where $C_l^{(i)}$ lies on the opposite side of $F$ to $C_k^{(i)}$.  But $\#(C_l^{(i)}) = \#(C_k^{(i)}) + 1$ since $C_l^{(i)}$ is on the same side of every $P_{ij^\prime}$ as $C_k^{(i)}$ except $P_{ij}$, and $C_l^{(i)} \subseteq H_{ij}$ while $C_k^{(i)} \nsubseteq H_{ij}$.  But this contradicts our choice of $C_k^{(i)}$ as one of the $C_1^{(i)}, \ldots, C_{m_i}^{(i)}$ contained in $U_i$ with maximal $\#(C_k^{(i)})$.  Therefore $C_k^{(i)} = V_i$ and the claim is proved.

{\em Claim 4: Each $V_i$ has non-zero measure. }
Suppose, in order to derive a contradiction, that $V_1$ has zero measure.  If $V_1, \ldots, V_k$ have zero measure and $V_{k+1}, \ldots, V_n$ have non-zero measure for some $k \ge 1$ then $V_i \subseteq U_i$ for all $i > k$ by Claim 3.  But $U_1$ meets each $U_j$ (for $j \not= 1$) in a set of zero measure, so $U_1$ also meets each $V_i$ in a set of zero measure when $i > k$.  Also, $V_1, \ldots, V_k$ all have zero measure so $U_1$ must meet them in sets of zero measure, too.  Hence $U_1 \subseteq \cX$ meets $\cup_{i=1}^n V_i = \cX$ in a set of zero measure, so $U_1$ has zero measure.  But this contradicts the fact that $U_1, \ldots, U_n$ is a partition, so the claim is proved.

{\em Claim 5: $V_i = U_i$. }
By Claims 3 and 4, $V_i \subseteq U_i$.  So by Claim 1, if some $V_i \not= U_i$ then there exists some $C_k^{(i)}$ which is contained in $U_i$ but meets $V_i$ in a set of measure $0$.  But $V_1, \ldots, V_n$ is a partition of $\cX$ so there is some $j \not= i$ so that $C_k^{(i)}$ meets $V_j$ in a set $C$ of non-zero measure.  But $V_j \subseteq U_j$ so $C$ lies in both $U_i$ and $U_j$, contradicting the fact that $U_1, \ldots, U_n$ is a partition and proving the claim and hence the Theorem.
\end{proof}

\begin{remark}
Note that Wallace's R1 \cite[p. 156]{wallace} implies that the partition of an SMML estimator consists of convex polytopes if and only if the stochastic family under consideration is an exponential family.
\end{remark}

\begin{remark}
Theorem \ref{T:Ui} also implies that each $U_i$ is the projection to $\R^d$ of one of the facets (i.e. $d$-dimensional faces) of the convex polytope
$$ \{ (x,y) \in \R^d \times \R \mid y \ge \lambda_i(x) \mbox{ for all $i=1,\ldots, n$} \}.$$
This description is useful in practice when trying to construct the partition corresponding to given assertions and coding probabilities.
\end{remark}

We also have the following corollary, which generalizes the one given in \cite{dowty} to higher dimensions.  Recall from Section \ref{S:review} that $q(x)$ and $\htheta(x)$ are step functions which are constant on the interior of each $U_i$ and are not defined on $\cup_{i=1}^n \partial U_i$, where two or more of the $U_i$ overlap.

\begin{corol}
$$ q(x) f(x | \htheta(x))  =  \max_{i \in \{1, \ldots, n\} } h(x) \exp (\lambda_i(x)) $$
for all $x$ in the dense subset of $\cX$ where the left-hand side is defined.  So even though $q(x) f(x | \htheta(x))$ is composed of step functions, it extends continuously to all of $\cX$.
\end{corol}

\begin{proof}
Note the left-hand side of the equation in the statement is defined exactly when $x$ lies in the interior of some $U_j$.  But in that case,
$$\log \left(q(x) f(x | \htheta(x))\right) = \log \left(q_j f(x | \htheta_j)\right)
= \log h(x) + \lambda_j(x) = \log h(x) + \max_{i \in \{1, \ldots, n\} } \lambda_i(x)$$
where the second equality used (\ref{E:exp_family}) and the last equality used  Theorem \ref{T:Ui}.
Since this formula holds for all $j$, taking exponentials completes the proof.
\end{proof}

\section{Inequalities obtained from the condition $\Idd \ge 0$}
\label{S:ineq}

In this section we use the condition $\Idd \ge 0$ on the second variation of $I_1$ to derive some novel inequalities.

Let $g$, $N$ and $D$ be as in Section \ref{S:deformations}, i.e., $D$ is a smooth $(d-1)$-dimensional, curvilinear disc contained in $U_1 \cap U_2$, $N$ is a smooth vector field which coincides on $D$ with the unit normal vector field to $D$ and $g : \R^d \to \R$ is a smooth function with support in $U_1 \cup U_2$.  Then we have the following general inequality.

\begin{lemma} \label{L:ineq2}  If an SMML estimator has partition $U_1, \ldots, U_n$, assertions $\htheta_1, \ldots, \htheta_n$ and coding probabilities $q_1, \ldots, q_n$ then
\begin{equation}\label{E:ineq2}
\| \htheta_1 - \htheta_2 \|  \ge \left[ \int_D  g^2 r  \, dx \right]^{-1} \left(
\left[ \frac{1}{q_1} + \frac{1}{q_2} \right] \left[\int_D g r \, dx\right]^2
+ \frac{1}{q_1} \delta_1 \cdot Q_1^{-1} \delta_1 + \frac{1}{q_2} \delta_2 \cdot Q_2^{-1} \delta_2 \right)
\end{equation}
for any $g$ and $D$ as above, where $\delta_i = \int_D (x-\mu(\htheta_i)) g r \, dx$ and $Q_i^{-1}$ is the Fisher information matrix for the expectation parameterisation evaluated at the expectation parameter $\mu(\htheta_i)$.
\end{lemma}

\begin{proof}[Proof of Lemma \ref{L:ineq2}]
An SMML estimator is a global minimum of $I_1$ so $\Id = 0$ and $\Idd \ge 0$ at $t=0$, and these relations must hold for all deformations so they hold for any $g$ and $D$.

By (\ref{E:Id}), $\Id = \int_D r g (\lambda_1 - \lambda_2) \, dx$, and this integral can vanish for all $g$ only if the integrand vanishes on $D$.  Since $r>0$, $\Id = 0$ therefore implies $\lambda_1 - \lambda_2$ vanishes on $D$.

Using this, (\ref{E:Idd}) and Remark \ref{R:Qi_is_FIM}, we then see that $\Idd \ge 0$ is equivalent to
\beq \label{E:Ntheta2}
\int_D  g^2 r N \cdot (\htheta_2 - \htheta_1) \, dx \ge \left[ \frac{1}{q_1} + \frac{1}{q_2} \right] \left[\int_D g r \, dx\right]^2
+ \frac{1}{q_1} \delta_1 \cdot Q_1^{-1} \delta_1 + \frac{1}{q_2} \delta_2 \cdot Q_2^{-1} \delta_2.
\eeq
The fact that $\lambda_1 - \lambda_2$ vanishes on $D$ implies that the unit normal $N$ to $D$ is proportional to $\grad(\lambda_1 - \lambda_2) = \htheta_1 - \htheta_2$, so $N \cdot (\htheta_2 - \htheta_1) = \pm \| \htheta_2 - \htheta_1 \|$.  It is not hard to see that the correct sign here is $+$, since $Q_1$ and $Q_2$ are positive definite by Theorem \ref{T:Id_Idd}, so all terms on the right-hand side of (\ref{E:Ntheta2}) are non-negative.  Hence $N \cdot (\htheta_2 - \htheta_1) = \| \htheta_2 - \htheta_1 \|$, so plugging this into (\ref{E:Ntheta2}) and rearranging proves the lemma.
\end{proof}

We now apply this general inequality to a particular case.  Suppose $\cX = \Theta = \R^d$ and that our exponential family consists of $d$-dimensional normal random variables with a variance-covariance matrix equal to the identity, i.e., $(X|\theta) \sim N_d(\theta,I_d)$.  Then $\psi(\theta) = \frac{1}{2} \| \theta \|^2$ so $\mu(\theta) = \theta$ and $Q_i = I_d$.  Let $\pi$ be the Jeffreys prior on a large region $R$ of $\R^d$, so that $\pi(\theta) = 1/V$ on $R$, where $V$ is the Euclidean volume of $R$, and $\pi(\theta) = 0$ outside of $R$.  The marginal distribution $r(x)$ is therefore approximately $r(x) = 1/V$ on $R$ and $r(x)=0$ outside of $R$, and this approximation is good except in a neighbourhood of the boundary of $R$.  So away from the boundary of $R$, we have the following result.

\begin{corol} For the exponential family $(X|\theta) \sim N_d(\theta,I_d)$ and the truncated Jeffreys prior described above, the partition and assertions of an SMML estimator must satisfy
\begin{equation} \label{E:ineq3}
\| \htheta_i - \htheta_j \| \ge \frac{A}{V_i} \left( 1 + \| \htheta_D - \htheta_i \| ^2 \right)
+   \frac{A}{V_j} \left( 1 + \| \htheta_D - \htheta_j \| ^2 \right)
\end{equation}
where $A$ is the (Euclidean) area of $U_i \cap U_j$, $V_i$ is the volume of $U_i$ and $\htheta_D$ is the centre of mass of $U_i \cap U_j$.  If the assertions form a lattice then writing $v = V_i/A$ we have
\beq \label{E:ineq4}
2 \le v \mbox{ and } v - \sqrt{v^2 - 4} \le \| \htheta_i - \htheta_j \| \le v + \sqrt{v^2 - 4}.
\eeq
\end{corol}

\begin{proof}
Take $D = U_i \cap U_j$ and $g(x)=1$ on $D$ in Lemma \ref{L:ineq2}.  Then $q_i = \int_{U_i} r(x) dx = V_i/V$  from (\ref{E:q}) and
$$\delta_i = \frac{1}{V} \int_D (x-\mu(\htheta_i)) \, dx
= \frac{1}{V} \int_D (x-\htheta_i) \, dx = \frac{A}{V}(\htheta_D - \htheta_i),$$
since $\mu(\theta) = \theta$.  Also, $Q_i$ is the identity matrix.  So plugging all of this into (\ref{E:ineq2}) proves (\ref{E:ineq3}).  If the assertions form a lattice then the partition is the corresponding Voronoi tessellation \cite[\S 3.3.4]{wallace} so $V_i = V_j$ and $\htheta_D$ lies in the plane equidistant from $\htheta_i$ and $\htheta_j$, hence
$$\| \htheta_D - \htheta_i \|^2
= \| \htheta_D - \frac{1}{2}(\htheta_i + \htheta_j) \|^2 + \| \htheta_i - \frac{1}{2}(\htheta_i + \htheta_j) \|^2
\ge \frac{1}{4}\| \htheta_i - \htheta_j \|^2.$$
Then (\ref{E:ineq4}) follows easily from (\ref{E:ineq3}).
\end{proof}

\section{Constructing SMML estimators}
\label{S:construction}

The usual approach to constructing an SMML estimator is to use (\ref{E:q}) and (\ref{E:htheta}) to, in effect, parameterize the assertions and coding probabilities by the partition and then to try to find the partition which minimizes the expression (\ref{E:I1_finitesum}) for $I_1$ \cite{wallace, farr, dowty}.  Theorem \ref{T:Ui} allows us to reverse this approach, i.e. to use the assertions and coding probabilities to parameterize the partition.  This is useful when $d \ge 2$ because then the set of all possible partitions is infinite dimensional while the assertions and coding probabilities are described by $n(d+1)$ numbers.  With this parametrization, (\ref{E:q}) and (\ref{E:htheta}) become $n(d+1)$ equations which are satisfied at the SMML estimator.  It is therefore possible to find the SMML estimator for a given number $n$ of regions by solving these equations.

In the case $d=1$, the above approach finds an SMML estimator by solving $2n$ equations in $2n$ unknowns while the approach of \cite{dowty} for the same problem solves $n-1$ equations in $n-1$ unknowns.  Therefore the method of \cite{dowty} is probably more efficient than the one above for $1$-dimensional problems.

Figure \ref{F:normal_normal} shows the partition of a `likely' SMML estimator for $2$-dimensional normal data with variance equal to the identity matrix and a normal prior for the mean, i.e. $(X|\theta) \sim N_2(\theta,I)$ and $\theta \sim N_2(0,4 I)$.  This was obtained by iteratively updating $\htheta_1, \ldots, \htheta_n$ and $q_1, \ldots, q_n$ (after starting at random values) by calculating $U_1, \ldots, U_n$ from Theorem \ref{T:Ui} and replacing each $q_i$ and $\htheta_i$ by the right-hand sides of (\ref{E:q}) and (\ref{E:htheta}), respectively.  This process was repeated until no $q_i$ and no coordinate of $\htheta_i$ changed by more than $10^{-5}$.  The resulting parameters and partition therefore approximately satisfy (\ref{E:q}), (\ref{E:htheta}) and (\ref{E:Ui}), so it is likely to be a good approximation to an SMML estimator.  Note, however, that while SMML estimators satisfy (\ref{E:q}), (\ref{E:htheta}) and (\ref{E:Ui}), the reverse might not be true, so we can only claim that this is a `likely' SMML estimator.  Note also that the symmetry of the normal-normal case considered above means that rotating an SMML estimator by any angle about the origin will give another SMML estimator, so there are infinitely many SMML estimators in this case.

\begin{figure}
  \centering
  \includegraphics[width=15cm]{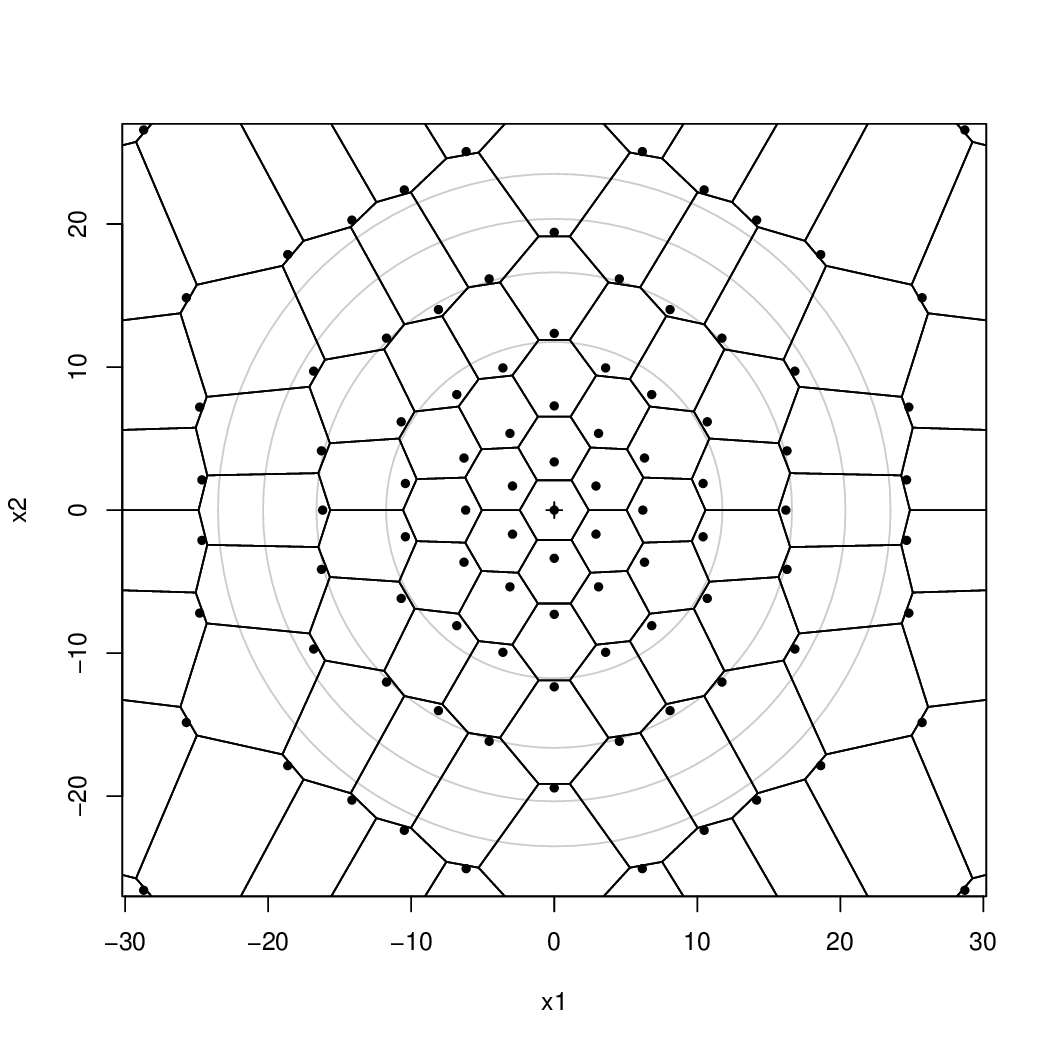} \\
  \caption{The partition of a likely SMML estimator for $2$-dimensional normal data with known variance and a normal prior for the mean.  Each dot shows the position of the centre of mass for the corresponding subset $U_i$.  The probability of a data point falling outside one of the grey circles is $10^{-m}$ where $m = 6,12,18,24$ in order of the innermost to outermost circle.}
  \label{F:normal_normal}
\end{figure}

\section{Summary}
\label{S:conclusion}

We studied SMML estimators for $d$-dimensional exponential families with continuous sufficient statistics.  Because the data space is continuous, we could use methods from calculus to study how the expected two-part code length $I_1$ changed under small deformations of the partition.  Since SMML estimators are global minima of $I_1$, all deformations of the partition of an SMML estimator must satisfy the conditions $\Id = 0$ and $\Idd \ge 0$ on the first and second variations of $I_1$.  These conditions were then used to prove that the partition of an SMML estimator consists of certain convex polytopes determined by its assertions and coding probabilities.  We also used the conditions $\Id = 0$ and $\Idd \ge 0$ to prove some novel inequalities for SMML estimators.  We further described a general method for constructing SMML estimators for exponential families, and we used this method to calculate an SMML estimator for a $2$-dimensional exponential family.  While the results given here apply for all $d$, this approach is probably less efficient than the one given in \cite{dowty} when $d=1$.

\appendix
\section{Proofs of technical lemmas}
\label{S:proofs}

We first prove the first and second variation formulae for $I_1$ under small deformations.

\begin{proof}[Proof of Theorem \ref{T:Id_Idd}]
Let $\hmu_i \defeq \mu(\htheta_i)$.  Then by (\ref{E:I1_finitesum}),
$$ C - I_1 = \sum_{i=1}^n \left(  q_i \left(\log q_i - \psi(\htheta_i) \right) +  \htheta_i \cdot q_i \hmu_i \right)$$
so differentiating both sides with respect to $t$ and denoting derivatives with dots gives
$$ -\Id = \sum_{i=1}^n \left(  \qd_i \left(\log q_i - \psi(\htheta_i)\right) + q_i \left(\frac{\qd_i}{q_i} - \frac{d}{dt}\psi(\htheta_i)\right)
          + \thetad_i \cdot q_i \hmu_i + \htheta_i \cdot \frac{d}{dt}\left(q_i \hmu_i\right) \right).$$
Now,
$$ \frac{d}{dt}\psi(\htheta_i) = \thetad_i \cdot \grad \psi |_{\htheta_i} = \thetad_i \cdot \hmu_i$$
by the chain rule and (\ref{E:gradpsi}), and $0 = \sum_{i=1}^n \qd_i$ since $1 = \sum_{i=1}^n q_i$ for all $t$.  Therefore
\begin{equation}\label{E:Id_general}
-\Id = \sum_{i=1}^n \left(  \qd_i \left(\log q_i - \psi(\htheta_i)\right) +
                       \htheta_i \cdot \frac{d}{dt}\left(q_i \hmu_i\right) \right).
\end{equation}
Differentiating this again and denoting second derivatives by double dots gives
\begin{equation}\label{E:Idd_general}
-\Idd = \sum_{i=1}^n \left(  \qdd_i \left(\log q_i - \psi(\htheta_i)\right) + \qd_i \left(\frac{\qd_i}{q_i} - \thetad_i \cdot \hmu_i \right) +
                       \thetad_i \cdot \frac{d}{dt}\left(q_i \hmu_i\right)
                       + \htheta_i \cdot \frac{d^2 \,}{dt^2}\left(q_i \hmu_i\right) \right).
\end{equation}

We now apply Lemma \ref{L:cdot} to calculate the derivatives of $q_i$ and $q_i \hmu_i$.  Setting $\rho = r$ in this lemma gives
\begin{eqnarray}
\qd_1 = - \qd_2 &=& \int_D g r \, dx  \label{E:qd} \\
\qdd_1 = - \qdd_2 &=& \int_D  g \nabla \cdot \left( g r N \right) \, dx \label{E:qdd}
\end{eqnarray}
where all the derivatives are evaluated at $t=0$.  Setting $\rho(x) = x_j r(x)$ in Lemma \ref{L:cdot}, where $x = (x_1, \ldots, x_d) \in \R^d$, gives us the first and second derivatives of the $j^{th}$ component of $q_i \hmu_i$.  Putting these components together gives
\begin{eqnarray}
\frac{d\,}{dt}\left(q_1 \hmu_1\right)  = - \frac{d\,}{dt}\left(q_2 \hmu_2\right) &=& \int_D x g r \, dx \label{E:qmud} \\
\frac{d^2\,}{dt^2}\left(q_1 \hmu_1\right) = - \frac{d^2\,}{dt^2}\left(q_2 \hmu_2\right) &=&
\int_D  \left(g^2 r N +  x g \nabla \cdot (g r N) \right) \, dx  \label{E:qmudd}
\end{eqnarray}
where, again, all the derivatives are evaluated at $t=0$.  Here we have used the fact that
$$ \nabla \cdot  \left( x_j g r N \right)
= g r N \cdot \grad x_j + x_j \nabla \cdot \left( g r N \right).$$
When $i \not= 1,2$, $U_i(t)= U_i$ for all $t$ so all derivatives of $q_i$ and $q_i \hmu_i$ vanish in this case.

Then combining (\ref{E:Id_general}), (\ref{E:qd}) and (\ref{E:qmud}) gives
\begin{eqnarray*}
-\left. \Id \right|_{t=0}
&=& \left(\int_D g r \, dx\right) \left(\log q_1 - \psi(\htheta_1)\right) + \htheta_1 \cdot \left(\int_D x g r \, dx\right) \\
&&  - \left(\int_D g r \, dx\right) \left(\log q_2 - \psi(\htheta_2)\right) - \htheta_2 \cdot \left(\int_D x g r \, dx\right) \\
&=& \int_D g r \left( \log q_1 - \psi(\htheta_1) + \htheta_1 \cdot x - \log q_2 + \psi(\htheta_2) - \htheta_2 \cdot x \right)\, dx \\
&=& \int_D r g (\lambda_1 - \lambda_2) \, dx
\end{eqnarray*}
as required, where $\lambda_i(x) = \log q_i + x \cdot \htheta_i - \psi(\htheta_i)$.

Substituting (\ref{E:qd})--(\ref{E:qmudd}) into (\ref{E:Idd_general}) gives
\begin{eqnarray*}
-\left. \Idd \right|_{t=0}
&=& \left(\int_D g \nabla \cdot \left( g r N \right) \, dx \right) \left(\log q_1 - \psi(\htheta_1) \right)
+ \frac{1}{q_1} \left(\int_D g r \, dx\right)^2  \\
&& - \left(\int_D g r \, dx\right) (\thetad_1 \cdot \hmu_1)
   + \thetad_1 \cdot \int_D x g r \, dx
   + \htheta_1 \cdot \int_D  \left( g^2 r N +  x g \nabla \cdot (g r N) \right) \, dx  \\
&& - \left(\int_D g \nabla \cdot \left( g r N \right) \, dx \right) \left(\log q_2 - \psi(\htheta_2) \right)
 + \frac{1}{q_2} \left(\int_D g r \, dx\right)^2  \\
&& + \left(\int_D g r \, dx\right) (\thetad_2 \cdot \hmu_2)
   - \thetad_2 \cdot \int_D x g r \, dx
   - \htheta_2 \cdot \int_D  \left( g^2 r N +  x g \nabla \cdot (g r N) \right) \, dx \\
&=& \int_D \left(\log q_1 - \psi(\htheta_1) + \htheta_1 \cdot x
    - \log q_2 + \psi(\htheta_2) - \htheta_2 \cdot x \right)
    g \nabla \cdot \left( g r N \right)  \, dx \\
&& + \left( \frac{1}{q_1} + \frac{1}{q_2} \right) \left(\int_D g r \, dx\right)^2  + \thetad_1 \cdot \int_D (x-\hmu_1) g r \, dx
  - \thetad_2 \cdot \int_D (x-\hmu_2) g r \, dx \\
&&   + \int_D  g^2 r N \cdot (\htheta_1 - \htheta_2) \, dx  \\
&=& \int_D (\lambda_1 - \lambda_2) g \nabla \cdot \left( g r N \right)  \, dx
   + \left( \frac{1}{q_1} + \frac{1}{q_2} \right) \left(\int_D g r \, dx\right)^2
   + \thetad_1 \cdot \delta_1 - \thetad_2 \cdot \delta_2 \\
&& + \int_D  g^2 r N \cdot (\htheta_1 - \htheta_2) \, dx.
\end{eqnarray*}
Now, $\hmu_i = \mu(\htheta_i)$ so $\mud_i = J_i \thetad_i$ where $J_i$ is the Jacobian matrix of $\mu$ evaluated at $\htheta_i$.  But by (\ref{E:gradpsi}) and (\ref{E:hesspsi}), $J_i=Q_i$ and $Q_i$ is symmetric and positive definite.  This implies that $Q_i$ is invertible so $\thetad_i = Q_i^{-1} \mud_i$.  Writing the left-hand side of (\ref{E:qmud}) as $\qd_1 \hmu_1 + q_1 \mud_1$, rearranging and using (\ref{E:qd}) gives
$$\mud_1 = \frac{1}{q_1} \left( \int_D x g r \, dx - \hmu_1 \int_D g r \, dx \right) = \frac{1}{q_1} \int_D (x - \hmu_1)g r \, dx = \frac{\delta_1}{q_1} $$
so $\thetad_1 = \frac{1}{q_1} Q_1^{-1} \delta_1$.  Similarly, $\thetad_2 = -\frac{1}{q_2} Q_2^{-1} \delta_2$, so the theorem follows.
\end{proof}

We cited the literature for a proof of Lemma \ref{L:cdot}.  However, this is a key lemma, so we also give a proof in this appendix, beginning with a general lemma.  See \cite{agricola} for an introduction to differential forms, Lie derivatives, Stokes' theorem, etc.

\begin{lemma}
\label{L:liederiv}
Let $\Omega \subseteq \R^d$ be a $d$-manifold with piecewise smooth boundary $\partial \Omega$.  Let $V$ be a vector field on $\R^d$ with flow $F_t: \R^d \to \R^d$ and let $\Omega_t \defeq F_t(\Omega)$.  If $c(t) \defeq \int_{\Omega_t} \omega$ for some differential $d$-form $\omega$ then
\begin{equation}\label{E:liederiv1}
\left. \frac{dc}{dt}\right|_{t=0} = \int_{\partial \Omega} i_V \omega
\end{equation}
and
\begin{equation}\label{E:liederiv2}
\left. \frac{d^2c}{dt^2}\right|_{t=0} = \int_{\partial \Omega} i_V \, d \, \, i_V \omega
\end{equation}
where $d$ is the exterior derivative and $i_V \alpha$ is the interior product of $V$ and any differential form $\alpha$.
\end{lemma}

\begin{proof}  We first note that
$$ \left. \frac{dc}{dt}\right|_{t=0} = \left. \frac{d\,}{dt}\right|_{t=0} \int_{\Omega_t} \omega
= \left. \frac{d\,}{dt}\right|_{t=0} \int_{\Omega} F_t^*\omega = \int_{\Omega} \left. \frac{d\,}{dt}\right|_{t=0} F_t^*\omega
= \int_{\Omega} \mathcal{L}_V \omega$$
where $\mathcal{L}_V$ is the Lie derivative of $V$.  So using Cartan's formula $\mathcal{L}_V = d \,\, i_V + i_V d$, the fact that $d\omega = 0$ (since $\omega$ is top-dimensional) and Stokes' theorem we have
$$ \left. \frac{dc}{dt}\right|_{t=0} = \int_{\Omega} (d \,\,i_V + i_V d) \omega = \int_{\Omega} d \,\,i_V \omega
= \int_{\partial \Omega} i_V \omega.$$
Similar reasoning gives
$$ \left. \frac{dc}{dt}\right|_{t} = \left. \frac{d\,}{ds}\right|_{s=0} c(s+t)
= \left. \frac{d\,}{ds}\right|_{s=0} \int_{\Omega} F_{s+t}^* \omega
= \int_{\Omega} \left. \frac{d\,}{ds}\right|_{s=0} F_s^*(F_t^*\omega)
= \int_{\Omega} \mathcal{L}_V (F_t^*\omega) = \int_{\partial \Omega} i_V (F_t^*\omega).$$
Now, if $v_1, \ldots, v_{d-1}$ are any vector fields on $\partial \Omega$ then
\begin{eqnarray*}
\left. \frac{d\,}{dt}\right|_{t=0} (i_V (F_t^*\omega))(v_1, \ldots, v_{d-1})
&=&  \left. \frac{d\,}{dt}\right|_{t=0} (F_t^*\omega)(V, v_1, \ldots, v_{d-1})  \\
&=&  (\mathcal{L}_V\omega)(V, v_1, \ldots, v_{d-1}) \\
&=&  (i_V \mathcal{L}_V\omega)(v_1, \ldots, v_{d-1})
\end{eqnarray*}
so
$$ \left. \frac{d\,}{dt}\right|_{t=0} i_V (F_t^*\omega)  = i_V \mathcal{L}_V \omega .$$
Therefore
$$\left. \frac{d^2c}{dt^2}\right|_{t=0}
= \left. \frac{d\,}{dt}\right|_{t=0} \int_{\partial \Omega} i_V (F_t^*\omega)
= \int_{\partial \Omega} i_V \mathcal{L}_V \omega
= \int_{\partial \Omega} i_V (d \,\, i_V + i_V d) \omega
= \int_{\partial \Omega} i_V \, d \, \, i_V \omega.$$
\end{proof}

We can now give our proof of Lemma \ref{L:cdot}.

\begin{proof} [Alternative proof of Lemma \ref{L:cdot}]
Let $x_1, \ldots, x_d$ be the standard co-ordinates on $\R^d$ so that $dx_1 \wedge \ldots \wedge dx_d$ is the $\R^d$ volume form.  We will apply Lemma \ref{L:liederiv} with $\Omega = U_i$ for $i$ either $1$ or $2$, $V=gN$ and $\omega = \rho dx_1 \wedge \ldots \wedge dx_d$.  Then (\ref{E:liederiv1}) becomes
$$ \left. \frac{dc}{dt}\right|_{t=0} = \int_{\partial \Omega} i_V \omega
= \int_{\partial U_i} i_{gN}(\rho dx_1 \wedge \ldots \wedge dx_d)
= \int_{\partial U_i} g \rho (i_N dx_1 \wedge \ldots \wedge dx_d)
= \pm \int_D g \rho dx$$
since $\mbox{Supp}(g) \cap \partial U_i \subseteq D$ and $i_N dx_1 \wedge \ldots \wedge dx_d$ is the volume form on $\partial U_1$ and minus the volume form on $\partial U_2$ (recall that $N$ is the unit normal vector field on $D$ which points out of $U_1$ and into $U_2$).

If $N = (N_1, \ldots, N_d)$ then
$$ i_V \omega =  \sum_{i=1}^d (-1)^{i+1} g \rho N_i dx_1 \wedge \ldots \wedge \widehat{dx_i} \wedge \ldots \wedge dx_d $$
where the hat indicates that the term is excluded.  Therefore
\begin{eqnarray*}
d \, \,i_V \omega
&=& \sum_{i,j=1}^d (-1)^{i+1} \frac{\partial \,}{\partial x_j} \left( g \rho N_i \right)
dx_j \wedge dx_1 \wedge \ldots \wedge \widehat{dx_i} \wedge \ldots \wedge dx_d \\
&=&  \sum_{i=1}^d (-1)^{i+1} \frac{\partial \,}{\partial x_i} \left( g \rho N_i \right)
dx_i \wedge dx_1 \wedge \ldots \wedge \widehat{dx_i} \wedge \ldots \wedge dx_d \\
&=&  \sum_{i=1}^d \frac{\partial \,}{\partial x_i} \left( g \rho N_i \right) dx_1 \wedge \ldots \wedge dx_d \\
&=&  \left(\nabla \cdot \left( g \rho N \right)\right) dx_1 \wedge \ldots \wedge dx_d
\end{eqnarray*}
so
\begin{eqnarray*}
i_V \, d \, \, i_V \omega
&=& \left(g \nabla \cdot \left( g \rho N \right)\right) \left(i_N dx_1 \wedge \ldots \wedge dx_d \right).
\end{eqnarray*}
Substituting this into (\ref{E:liederiv2}) then completes the proof of the lemma since $g$ is $0$ on $\partial U_i$ except perhaps in $D$ and $i_N dx_1 \wedge \ldots \wedge dx_d$ is the volume form on $\partial U_1$ and minus the volume form on $\partial U_2$.
\end{proof}

\end{document}